\documentclass[orivec,runningheads]{llncs}
\usepackage[T1]{fontenc}
\usepackage{bm}        
\usepackage{amssymb}   
\usepackage{amsfonts}
\usepackage{float}
\usepackage{amsmath}
\usepackage{amssymb, amsfonts}

\usepackage{graphicx}    
\usepackage{subcaption}   

\DeclareMathOperator{\prox}{prox}
\DeclareMathOperator*{\argmin}{argmin}
\usepackage{color}
\begin{document}
\author{Na\"il Khelifa\inst{1}\thanks{Work completed during a research visit to DAMTP at the University of Cambridge.}\and Carola-Bibiane Sch\"onlieb\inst{2}\and Ferdia Sherry\inst{2}}
\institute{ENSAE Paris \email{nail.khelifa@ensae.fr}\and DAMTP, University of Cambridge \email{\{cbs31,fs436\}@cam.ac.uk}}
\title{Enhanced Denoising and Convergent Regularisation Using Tweedie Scaling}
\titlerunning{Convergent regularisation in PnP via Tweedie scaling}
%
\maketitle              
\begin{abstract}

The inherent ill-posed nature of image reconstruction problems, due to limitations in the physical acquisition process, is typically addressed by introducing a regularisation term that incorporates prior knowledge about the underlying image. The iterative framework of Plug-and-Play methods, specifically designed for tackling such inverse problems, achieves state-of-the-art performance by replacing the regularisation with a generic denoiser, which may be parametrised by a neural network architecture. However, these deep learning approaches suffer from a critical limitation: the absence of a control parameter to modulate the regularisation strength, which complicates the design of a convergent regularisation. To address this issue, this work introduces a novel scaling method that explicitly integrates and adjusts the strength of regularisation. The scaling parameter enhances interpretability by reflecting the quality of the denoiser's learning process, and also systematically improves its optimisation. Furthermore, the proposed approach ensures that the resulting family of regularisations is provably stable and convergent.

\keywords{Image Reconstruction  \and Plug-and-Play \and Convergent Regularisation.}
\end{abstract}
\vspace{-3em}
\section{Introduction}
\vspace{-0.5em}

\subsubsection{Image reconstruction.}
In various fields, physical devices capture real-world scenes and convert them into digital form. The image reconstruction (IR) problem involves reversing this process by estimating a realistic input image that best matches the observed data, which is typically a degraded version of the original scene due to limitations in the acquisition system. Mathematically, if $\mathcal{X} = \mathbb{R}^n$ denotes the space of images and $\mathcal{Y} = \mathbb{R}^m$ that of measurements, the physics behind a measurement $y \in \mathcal{Y}$ and its original image $x \in \mathcal{X}$ is commonly modeled by a linear operation, 
\begin{equation}
    y = \mathrm{A}x + \xi,
    \label{physics-general-formula}
\end{equation}
where $\mathrm{A} : \mathcal{X} \rightarrow \mathcal{Y}$ is the \emph{forward operator} and $\xi$ is an additive white Gaussian noise (AWGN). The reconstruction problem writes as the minimisation over $\mathcal{X}$ of a data-fidelity term, often and here taken to be $\frac{1}{2}\|\mathrm{A}x - y\|^2$. However, this problem is commonly \textit{ill-posed}, meaning that its solution is either (or both) not unique or highly sensitive to changes in the initial data. To address this issue, a standard approach is to incorporate regularisation by reformulating the IR problem as
\begin{equation}
    \underset{x \in \mathcal{X}}{\min } ~~ \Big\{ \frac{1}{2}\|\mathrm{A}x - y\|^2 + \lambda g(x) \Big\},
    \label{regularised-version}
\end{equation}
where $g : \mathcal{X} \to \mathbb{R}$ is the regularisation term and $\lambda \geq 0 $ is a term controlling its intensity. By viewing $x$ as the realization of a random variable $X$ with prior distribution $p_X$ and $y$ as the realization of a random variable $Y$ with a conditional density $p_{Y|X}$, the \textit{maximum a posteriori} (MAP) estimate is defined by,
\begin{equation}
    \max_{x \in \mathcal{X}} \, p_{X|Y}(x|y) = \min_{x \in \mathcal{X}} \Big\{ -\log p_{Y|X}(y|x) - \log p_X(x) \Big\}.
    \label{MAP-bayesian-setting}
\end{equation}
Identifying problems \eqref{regularised-version} and \eqref{MAP-bayesian-setting}, the MAP links the data-fidelity term to the likelihood and the regularisation to the prior.
\vspace{-1.5em}
\subsubsection{Plug-and-Play.}
In most cases, problem \eqref{regularised-version} does not have an analytical solution, necessitating the use of iterative optimisation algorithms.  First-order proximal splitting methods \cite{pesquet-combettes-2011} such as Proximal Gradient Descent (PGD, \cite{PGD}), Half-Quadratic Splitting (HQS, \cite{HQS}) and Alternating Direction Method of Multipliers (ADMM, \cite{ADMM}) are a common tool for minimising the sum of two functions. These methods use the proximal operator of a functional $f$, defined as,
$$
\prox_f(x) = \argmin_{z} \left\{ \frac{1}{2} \|x - z\|^2 + f(z) \right\}.
$$
In the context of IR, the Plug-and-Play (PnP) framework \cite{venkatakrishnan-2011} extends these methods by substituting the proximal operator of the regularisation term with a generic Gaussian denoiser, i.e.\ an operator designed to solve the Gaussian denoising task that corresponds to the physics \eqref{physics-general-formula} with $\mathrm{A} = \mathrm{I}_n$ and $\xi$ an AWGN. Initially applied with non-deep denoisers such as BM3D \cite{chan-2016} or pseudo-linear denoisers \cite{gavaskar-2021,nair-2021}, more recent approaches \cite{meinhardt-2017,sun-2019a,ahmad-2020,yuan-2020,sun-2021} rely on denoisers parametrised with deep convolutional neural networks such as DnCNN \cite{zhang-2017a}, IR-CNN \cite{zhang-2017b} or DRUNet \cite{DRUNet}. 
The PnP framework has been applied in conjunction with first-order proximal splitting algorithms yielding PnP-HQS \cite{zhang-2017b}, PnP-ADMM \cite{romano-2017,ryu-2019}, and PnP-PGD \cite{PnP-with-PGD}.
\vspace{-1.5em}
\subsubsection{Problem.}While these deep learning methods yield impressive restoration results, they come with notable limitations. In particular, their interpretation as regularisation terms becomes problematic within the framework of convergent regularisation. These models are typically trained at fixed noise levels, yet their penalising role in a PnP setting, along with the requirement for convergent regularisation, calls for adjustable control over the regularisation strength. Thus, it is desirable to devise a general approach that not only introduces an interpretable intensity control to a pre-trained denoiser but that also leads to convergent regularisation --- the method introduced in \cite{denoiser-boosting} does not, see Section \ref{sec:empirical_conv_reg}.
\vspace{-1.5em}
\subsubsection{Related Works.} In \cite{denoiser-boosting}, the authors propose a methodology for adjusting the regularisation strength of a denoiser by applying a scaling, which is motivated by interpreting the denoiser as the proximal operator of a 1-homogeneous functional. In that work, the scaling is shown to improve IR performance, but the question of convergent regularisation is not studied, and in fact, as we show later in this work, it does not generally give convergent regularisation. The problem of showing convergent regularisation for PnP methods has received some attention recently \cite{haltmeier-ebner,hauptmannConvergentRegularizationInverse2024}. In particular, the work in \cite{haltmeier-ebner} introduces a theoretical framework for constructing a convergent regularisation for the PnP-PGD; however, the authors do not propose a concrete method to implement this construction with generic denoisers. Let us take the chance to contrast this line of work with a different but related line of work on proving that PnP methods define convergent iterations, a topic which has received considerably more attention from various perspectives recently \cite{ryu-2019,nair-2021,HuraultSSVM}: having convergent iterations is necessary, but not sufficient, to make sense of convergent regularisation in our setting.
\vspace{-1.5em}
\subsubsection{Main Contributions.}
This paper bridges the gap between the two contributions in \cite{denoiser-boosting} and \cite{haltmeier-ebner} by introducing a novel framework for constructing convergent regularisation through a scaling transformation of a single pre-trained denoiser, aligning with the setting proposed in \cite{haltmeier-ebner}. Additionally, it is shown that the introduced scale parameter has a highly interesting interpretation. Indeed, a specific value of the scale parameter consistently enhances the performance of deep denoisers trained with the $\mathcal{L}^2$ loss and acts as a quality indicator for the training. 
\vspace{-2em}
\subsubsection{Notations.} Throughout this work, for simplicity and seeing images and measurements as high-dimensional tensors, let $\mathcal{X} = \mathbb{R}^n$ and $ \mathcal{Y} = \mathbb{R}^m$. The noise level is denoted $\sigma > 0$. Let $\xi$ be a random variable distributed according to $\mathcal{N}(0, \mathrm{I}_n)$ and let $\mathbf{D} : \mathcal{Y} \rightarrow \mathcal{X}$ denote a Gaussian denoiser trained to minimise the $\mathcal{L}^2$ loss at a noise level $\sigma$. Let also \( p_X \) denote the prior distribution of clean images over \( \mathcal{X} \), \( G_\sigma \) the density of the Gaussian distribution \(\mathcal{N}(0, \sigma^2 \mathrm{I}_n)\), and \( p_\sigma = p_X * G_\sigma \) the distribution of noisy images at noise level \(\sigma\) over \( \mathcal{Y} \). The minimum mean square error (MMSE) denoiser, is denoted and defined by
\begin{equation*}
\mathbf{D}^{\text{MMSE}} :=\argmin_{{\mathbf{D} \text{ measurable}}}\mathbb{E}_{X, Y}[\|\mathbf{D}(Y) - X \|^2] = \mathbb{E}[X|Y].
\end{equation*}
For any denoiser $\mathbf{D}$, the mean squared error of Gaussian denoising is denoted,
$$
\mathcal{L}^2(\mathbf{D}) = \mathbb{E}_{X, \xi}[\|\mathbf{D}(X + \sigma \xi) - X\|^2].
$$
\vspace{-3em}
\subsubsection{Outline.} The first section presents the proposed scaling-based regularisation method, providing a justification for its theoretical foundation and underlying intuition. The second section explores the interpretation of the scaling parameter that governs the introduced regularisation. Finally, the last section demonstrates that the resulting family of regularisations forms a convergent regularisation.

\section{Construction}

In this section, the aim is twofold: introduce a scaling-based control parameter for a pre-trained denoiser and provide insights into the meaning of the parameter.
\vspace{-2em}
\subsubsection{Homogeneous scaling.} In \cite{denoiser-boosting}, the scaling approach is motivated by the observation that, if $h : \mathcal{X} \rightarrow \mathbb{R}$ is a 1-homogeneous function (i.e.\ $h(\tau u) = \tau h(u)$ for all $\tau > 0, u \in \mathcal{X}$) whose proximal operator is well-defined, then for all $x \in \mathcal{X}$,
\begin{equation*}
    \prox_{\tau h} (x) = \tau \prox_h \left(\frac{x}{\tau}\right).
    \label{prox-scaling-formula}
\end{equation*}
This implies that if the initial denoiser is given by $\mathbf{D} = \operatorname{prox}_h$ for some 1-homogeneous $h$, then, for all $\delta > 0$ and for all $x \in \mathcal{X}$,
\begin{equation*}
    \frac{1}{\delta}\mathbf{D}(\delta x)
    = \frac{1}{\delta} \operatorname{prox}_h(\delta x) = \operatorname{prox}_\frac{h}{\delta}(x).
\end{equation*}
From this approach, henceforth referred to as \textit{homogeneous scaling}, we aim to retain the concept of a scaling transformation while discarding the strong assumption on the proximal expression of $\mathbf{D}$.
\vspace{-1.5em}
\subsubsection{Proposed Approach.} The proposed approach builds upon Tweedie’s identity \cite{efron-2011}, which states,  
\begin{equation}
\mathbf{D}^\text{MMSE} = \operatorname{Id} + \sigma^2 \nabla \log p_\sigma.
\label{tweedie-identity}
\end{equation}
The objective is to define a method ensuring convergent regularisation. This will be achieved by studying the behaviour of the identity as \( \sigma \) approaches zero. The idea here is to provide an intuition about the proposed method. Assuming \( p_X \) is sufficiently smooth ($\mathcal{C}^3$ suffices, for example), \( p_\sigma \) satisfies the heat equation with time variable \( t = \sigma^2 \),
\begin{align}\partial_{\sigma^2} p_\sigma = \mathrm{\Delta} p_\sigma, \qquad \partial_{\sigma^2} \nabla p_\sigma = \nabla (\mathrm{\Delta} p_\sigma).
    \label{heat-eq}
\end{align}

Using the first-order Taylor expansion of \( p_\sigma \) and \( \nabla p_\sigma \) around \( \sigma^2 = 0 \), and incorporating \eqref{heat-eq}, the following expressions are obtained,
\begin{align*}
\begin{cases}
    p_\sigma = p_X + \partial_{\sigma^2} p_\sigma \big|_{\sigma^2=0} \sigma^2 + \mathcal{O}_0(\sigma^4) = p_X + \sigma^2 \mathrm{\Delta} p_X + \mathcal{O}_0(\sigma^4), \\
    \\
    \nabla p_\sigma = \nabla p_X + \partial_{\sigma^2} \nabla p_\sigma \big|_{\sigma^2=0} \sigma^2 + \mathcal{O}_0(\sigma^4) = \nabla p_X + \sigma^2 \nabla \mathrm{\Delta} p_X + \mathcal{O}_0(\sigma^4).
\end{cases}
\end{align*}

As \( \sigma \to 0 \), neglecting higher-order terms in the expansion yields,
\begin{align*}
    \nabla \log p_\sigma = \frac{\nabla p_\sigma}{p_\sigma} = \frac{\nabla p_X + \sigma^2 \nabla \mathrm{\Delta} p_X}{p_X + \sigma^2 \mathrm{\Delta} p_X} = \frac{\nabla p_X}{p_X} + \mathcal{O}_0(\sigma^2) \sim_{\sigma \to 0} \nabla \log p_X.
\end{align*}

Thus, for Tweedie’s identity, it follows that,
\begin{align*}
    \operatorname{Id} + \sigma^2 \nabla \log p_\sigma \sim_{\sigma \to 0} \operatorname{Id} + \sigma^2 \nabla \log p_X.
\end{align*}

This leads to the introduction of a scaling parameter to control regularisation, expressed in Tweedie’s identity as,
\begin{align*}
    \operatorname{Id} + \frac{\sigma^2}{\delta^2} \nabla \log p_X =  \operatorname{Id} + \frac{1}{\delta^2} \left( \operatorname{Id} + \sigma^2 \nabla \log p_X - \operatorname{Id} \right)
    \sim_{\sigma \to 0} \operatorname{Id} + \frac{\mathbf{D}^\text{MMSE} - \operatorname{Id}}{\delta^2}.
\end{align*}
The key observation is that training a denoiser \( \mathbf{D} \) to minimise the \( \mathcal{L}^2 \)-loss for a noise level \( \sigma \) effectively approximates \( \mathbf{D}^\text{MMSE} \). By substituting \( \mathbf{D} \) for \( \mathbf{D}^\text{MMSE} \) in the above expression, the \textit{Tweedie Scaling Method} is defined as,
\begin{equation*}
\mathbf{D}_\delta = \operatorname{Id} + \frac{\mathbf{D} - \operatorname{Id}}{\delta^2}.
\end{equation*}
This formulation assumes no specific structure for \( \mathbf{D} \) and relies on two key approximations: (1) the approximation of \( \mathbf{D} \) as \( \mathbf{D}^\text{MMSE} \) through proper \( \mathcal{L}^2 \)-training, and (2) the validity of expansions as \( \sigma \to 0 \). In fact, as we will show in Section \ref{sec:conv_reg_prop}, this scaling approach can result in convergent regularisation even in a setting in which we do not assume that the denoiser is close to being optimal.

\vspace{-1em}
\section{Enhanced denoising and interpretation of $\delta$}
\label{sec:interpretation_scaling}
\vspace{-1em}
To interpret \( \delta \), recall that the construction of \( \mathbf{D}_\delta \) in the previous section started with an approximation of \( \mathbf{D}^\text{MMSE} \) by \( \mathbf{D} \), where the accuracy was related to the performance of \( \mathbf{D} \) in minimizing the \( \mathcal{L}^2 \) loss, or equivalently, to the tightness of the following bound --- which holds by definition of \( \mathbf{D}^\text{MMSE} \):
\begin{equation}
    \mathcal{L}^2(\mathbf{D}_\text{MMSE}) \leq \mathcal{L}^2(\mathbf{D})
    \label{l2-loss-mmse}
\end{equation}
To proceed with the interpretation, this bound is refined by the following result,
\begin{proposition}
    For 
\begin{align*}
    \delta^2_\mathrm{opt} = - \frac{\mathbb{E}[\|\mathbf{D}(X + \sigma \xi) - (X + \sigma \xi)\|^2]}{\mathbb{E}[\langle \sigma \xi, \mathbf{D}(X + \sigma \xi) - (X + \sigma \xi) \rangle]},
\end{align*}
it holds that,
\begin{equation}
    \mathcal{L}^2(\mathbf{D}_\mathrm{MMSE}) \leq \mathcal{L}^2(\mathbf{D}_{\delta_\mathrm{opt}}) \leq \mathcal{L}^2(\mathbf{D})
    \label{improving-formulation}
\end{equation}
\end{proposition}
\begin{proof}
    The proof directly follows from studying the function \( \delta \mapsto \mathcal{L}^2(\mathbf{D}_\delta) \).
\end{proof}
Now, If equality holds in \eqref{l2-loss-mmse}, perfect training of $\mathbf{D}$ follows, implying:
\begin{align*}
    \mathcal{L}^2(\mathbf{D}_\text{MMSE}) = \mathcal{L}^2(\mathbf{D}_{\delta_\text{opt}}) = \mathcal{L}^2(\mathbf{D}) \implies \mathbf{D}_{\delta_\text{opt}} = \mathbf{D} \iff \delta_\text{opt}^2 = 1
\end{align*}
Thus, the value of \( \delta_\text{opt} \) reveals the quality of the initial denoiser’s training. It is noteworthy that, when \( \mathbf{D} \) is not perfectly trained, it is possible to construct a denoiser with a lower theoretical training error simply by choosing the appropriate value of the transformation parameter.

\section{Plug-and-play properties of $(\mathbf{D}_\delta)_\delta$} 
\vspace{-1em}
Attention is now directed towards the behaviour of $(\mathbf{D}_\delta)_\delta$ within a Plug-and-Play framework. Specifically, it is first demonstrated that the Tweedie scaling method (1) defines convergent iterations for PnP-PGD algorithms under mild assumptions and (2) establishes a convergent regularisation for the original problem, again under mild assumptions.
\vspace{-1.5em}
\subsection{Convergence of iterations}
The focus is on the convergence of iterations within the PnP-PGD \cite{PnP-with-PGD} algorithm which, for a mean square data-fidelity term and considering $\mathbf{D}_\delta$, iterates as:
\begin{align}
    x^{(i+1)} = \mathbf{D}_\delta(x^{(i)} - \tau \mathrm{A}^\top(\mathrm{A}x^{(i)} - y)) := \mathbf{D}_\delta(\mathbf{G}(x^{(i)})) := \mathbf{T}_\delta(x^{(i)})
    \label{PnP-PGD-iterations}
\end{align}
where $\mathbf{G}(x) := x - \tau \mathrm{A}^\top (\mathrm{A}x - y)$ denotes the gradient-step operator and $\mathbf{T}_\delta(x) := \left(1 - 1/\delta^2\right)\mathbf{G}(x) + (\mathbf{D} \circ \mathbf{G})(x)/\delta^2$. In this setting, we have the following result:

\begin{proposition}[$\mathbf{T}_\delta$ defines convergent iterations]
    Assume that the step size $\tau$ is chosen with $0 < \tau \leq 1/\|\mathrm{A}^\top  \mathrm{A}\|$ and that $\mathbf{D}$ is non-expansive. Then,  $\mathbf{T}_\delta$ is $\frac{\delta^2}{2\delta^2 - 1}$-averaged. If $\mathbf{T}_\delta$ has at least one fixed point and $\delta^2 > 1$, then the iterations ~\eqref{PnP-PGD-iterations} converge to a fixed point of $\mathbf{T}_\delta$.
    \label{convergence-of-iterations}
\end{proposition}

\begin{proof}
The proof follows in three steps: 
\begin{enumerate}
    \item Recalling that an operator \(\mathbf{O}\) is said to be \(\theta\)-averaged for \(\theta \in (0, 1)\) if it can be expressed as $\mathbf{O} = \theta \mathbf{R} + (1 - \theta) \operatorname{Id}$, where \(\mathbf{R}\) is non-expansive, and recalling that $\mathbf{D}_\delta = \left(1 - 1/\delta^2\right)\operatorname{Id} + \mathbf{D}/\delta^2$, one directly gets that, assuming $\mathbf{D}$ to be non-expansive, implies that $\mathbf{D}_\delta$ is $1/\delta^2$-averaged.
    \item Moreover, it is recalled from \cite{bauschke-combettes} that  if \(\mathbf{O}_1\) and \(\mathbf{O}_2\) are \(\theta_1\)-averaged and \(\theta_2\)-averaged operators, respectively, then their composition \(\mathbf{O}_1 \circ \mathbf{O}_2\) is \(\theta\)-averaged with
    \[
    \theta = \frac{\theta_1 + \theta_2 - 2\theta_1\theta_2}{1 - \theta_1\theta_2}.
    \]
    \noindent
    Combining this, the $1/\delta^2$-averagedness of $\mathbf{D}_\delta$ and, the result from \cite{bauschke-combettes} that states that under the stepsize condition \(0 < \tau \leq 1/\|\mathrm{A}^\top \mathrm{A}\|\), \(\mathbf{G}\) is firmly non-expansive, i.e.\ $1/2$-averaged and the fact that  \(\mathbf{T}_\delta = \mathbf{D}_\delta \circ \mathbf{G}\) it follows that \(\mathbf{T}_\delta\) is \(\delta^2/(2\delta^2 - 1)\)-averaged. 
    \item According to a corollary of the Krasnosel'skii-Mann theorem \cite{bauschke-combettes}, if \(\mathbf{T}_\delta\) admits at least one fixed point and if \(\delta^2/(2\delta^2 - 1) \in (0, 1)\)—a condition equivalent to \(\delta^2 > 1\)—then the sequence of iterations \eqref{PnP-PGD-iterations} converges to a fixed point.
\end{enumerate}
\end{proof}

The central assumption underlying this result is the non-expansiveness of \(\mathbf{D}\), which is influenced by the architecture of \(\mathbf{D}\). In the experiments, we will use a standard neural network denoiser, the DRUNet (see Section \ref{sec:experiments} for more details). As demonstrated in Table \ref{empirical-test-of-non-expansiveness}, the maximum value of \(\|\mathbf{D}(y_1) - \mathbf{D}(y_2)\|/\|y_1 - y_2\|\) for two images $y_1, y_2$ from the (noisy) CBSD68 patch-extend dataset is below one for sufficiently large noise levels, supporting the assumption of non-expansiveness.
\vspace{-0.5cm}
\begin{table}[ht]
    \caption{\centering $\max\{\frac{\|\mathbf{D}(y_1) - \mathbf{D}(y_2)\|_2}{\|y_1 - y_2\|_2} | y_1, y_2 \in \mathcal{D}_{\mathrm{CBSD68}}(\sigma^2)\}$ as the noise level $\sigma^2$ varies.}\label{empirical-test-of-non-expansiveness}
    \centering
    \begin{tabular}{|c|c|c|c|c|c|c|c|c|c|}
        \hline
         & $\sigma = 0.01$ & $\sigma = 0.03$ & $\sigma = 0.05$ & $\sigma = 0.1$ & $\sigma = 0.15$ & $\sigma = 0.2$ & $\sigma = 0.25$ & $\sigma = 0.3$\\
         \hline
        $\text{max}$ & $1.005$& $0.997$ & $0.994$ & $0.985$ & $0.966$ & $0.950$ & $0.919$ & $0.905$\\ 
        \hline
    \end{tabular}
\end{table}
\vspace{-1.5em}
\subsection{Convergent Regularisation Properties}
\label{sec:conv_reg_prop}
The properties of convergent regularisation are now considered. In \cite{haltmeier-ebner}, a theoretical framework for studying convergent regularisation in PnP methods is introduced, but no general method is provided for constructing it from a pre-trained denoiser \(\mathbf{D}\). This work demonstrates that the proposed Tweedie scaling method defines a convergent regularisation that satisfies the framework in \cite{haltmeier-ebner}.\\
For $\delta > 0$ and noisy measurements $y$, the set of fixed points of the iterations of the PnP-PGD algorithm is defined as follows,
\begin{equation}
    \operatorname{PnP}(\delta, y) := \operatorname{Fix}(\mathbf{D}_\delta \circ (\mathrm{Id} - \tau \mathrm{A}^\top (\mathrm{A}(\cdot) - y))),
\end{equation}
where $\operatorname{Fix}(\mathbf{T}) = \{x \in \mathcal{X} ~ | ~ \mathbf{T}(x) = x\}$. By adapting the result of Definition 2.5 in \cite{haltmeier-ebner} to our case, $(\mathbf{D}_\delta)_\delta$  define a convergent regularisation  for the problem $\argmin_{x} \frac{1}{2}\|\mathrm Ax - y \|^2$ over $E \subseteq \mathcal{X}$, if the following conditions are satisfied,
\begin{itemize}
    \item[(a)] \textbf{Stability:} For all $\delta > 0$, $\mathbf{D}_\delta$ is continuous.
    \item[(b)] \textbf{Convergence:} For all $y \in E$, $(1/\delta_k^2)_{k \in \mathbb{N}} \subset (0,\infty)^{\mathbb{N}}$ converging to $0$ and all $y_k \in \mathcal{Y}$ with $\|y - y_k\| \leq 1/\delta_k^2$, the sequence $(\operatorname{PnP}(1/\delta_k^2, y_k))_{k \in \mathbb{N}}$ has a convergent subsequence and the limit of every convergent subsequence of $(\operatorname{PnP}(1/\delta_k^2, y_k))_{k \in \mathbb{N}}$ is a minimiser of $x \mapsto\frac{1}{2}\|\mathrm{A}x - y\|^2$.
\end{itemize}

\vspace{-1.5em}
\subsubsection{Stability}
To demonstrate stability, Theorem 3.10 and Corollary 3.11 from \cite{haltmeier-ebner} are considered. In the statement of these results, four conditions are imposed: a contractivity condition on $\mathbf{D}_\delta$ and three conditions on the data-fidelity function. It has been shown in \cite{haltmeier-ebner} that the data-fidelity between a clean image $x$ and its noisy measurements $y$  defined by $\frac{1}{2}\|\mathrm{A}x - y\|^2$ satisfies these conditions. Thus, the only assumption that remains to be verified is the contractivity of $\mathbf{D}_\delta$ for all $\delta>0$, which is equivalent to the contractivity of the original denoiser $\mathbf{D}$. By directly applying Corollary 3.11 from \cite{haltmeier-ebner}, it follows that
\begin{equation*}
    \text{\textit{Assuming that} }\mathbf{D}_\delta\text{ \textit{is contractive,} }x \mapsto \operatorname{PnP}(\delta, x)\text{ \textit{is continuous for all} }\delta > 0.
\end{equation*}
\vspace{-3.5em}
\subsubsection{Convergence}
To prove convergence, Theorem 3.14 from \cite{haltmeier-ebner} is applied, requiring that $(\mathbf{D}_\delta)_\delta$ be an admissible family of denoisers, defined by,
\begin{center}
\begin{itemize}
    \item[] $ \quad$ \textbf{(A1)} $\quad \forall \delta > 0$, $\mathbf{D}_\delta(\cdot): \mathbb{R}^n \rightarrow \mathbb{R}^n$ is a contraction.
    \item[] $ \quad$ \textbf{(A2)} $\quad \mathbf{D}_\delta(\cdot) \longrightarrow \mathrm{Id}_{\mathbb{R}^n}$ strongly, pointwise as $\delta \to +\infty$.
    \item[] $ \quad$ \textbf{(A3)} $\quad \mathbf{D}_\delta(\cdot) \longrightarrow \mathrm{Id}_{\mathbb{R}^n}$ weakly, uniformly over bounded sets as $\delta \to +\infty$.
    \item[] $ \quad$ \textbf{(A4)} $\quad \exists E \subseteq \mathbb{R}^n \, \forall x \in E: \|\mathbf{D}_\delta(x) - x\| = \mathcal{O}(1 - \operatorname{Lip}(\mathbf{D}_\delta)).$
\end{itemize}
\end{center}

The first condition is overly restrictive regarding the architecture defining $\mathbf{D}$, and it is desirable to move beyond it. In \cite{haltmeier-ebner}, a rescaling method is proposed to relax the contractivity condition \textbf{(A1)} to non-expansiveness. More precisely, it is suggested that if the family is non-expansive and satisfies \textbf{(A2)} and \textbf{(A3)}, then it is always possible to construct a family of admissible denoisers by considering \( (\gamma(\frac{1}{\delta^2}) \mathbf{D}_\delta)_{\delta > 0} \), where \( \gamma : [0, \infty) \to (0, 1] \) is a strictly decreasing and continuous function at zero, satisfying \( \gamma(0) = 1 \), and where \( \|\mathbf{D}_\delta(x) - x\| = \mathcal{O}(1 - \gamma(\frac{1}{\delta^2})) \) as \( \frac{1}{\delta^2} \to 0 \). For instance, a valid choice for \( \gamma \) is \(\gamma\left(1/\delta^2\right) = \delta^2/(1 + \delta^2)\). It is noted in \cite{haltmeier-ebner} that if conditions \textbf{(A2)} and \textbf{(A3)} are satisfied by \((\mathbf{D}_\delta)_{\delta > 0}\), then they are also satisfied by \((\gamma(\frac{1}{\delta^2}) \mathbf{D}_\delta)_{\delta > 0}\), and in fact this rescaled family then satisfies \textbf{(A4)}. Consequently, in what follows, by a slight abuse of notation, \(\mathbf{D}_\delta\) will be identified with this rescaled version. Thus, only conditions \textbf{(A2)} and \textbf{(A3)} need to be verified, since it was shown in the previous section that condition \textbf{(A1)}, now relaxed to non-expansiveness, is empirically satisfied. Since we have assumed to be in a finite-dimensional setting, weak and strong convergence coincide, so it is clear that \textbf{(A3)} implies \textbf{(A2)}, and we only need to prove \textbf{(A3)}:\\

\noindent
\textbf{(A3)}: Let \( B \subseteq \mathbb{R}^n \) be a bounded set with $R = \sup_{x \in B} \|x\|$. For any \( x \in B \),

\[
\| \mathbf{D}_\delta (x) - x \| = \left\| \frac{\mathbf{D}(x) - x}{\delta^2} \right\|
\leq \frac{\|\mathbf{D} - \operatorname{Id} \|}{\delta^2} \| x \| \leq \frac{R \|\mathbf{D} - \operatorname{Id} \|}{\delta^2}
\]
Now, the bound on the right-hand side is uniform over all \( x \in B\) and tends to zero as \( \delta \to \infty \), proving uniform convergence on bounded sets.\\

\noindent
Hence, $(\mathbf{D}_\delta)_\delta$ is admissible, completing the proof that the proposed family defines a convergent regularisation.
\vspace{-1em}
\section{Experiments}
\label{sec:experiments}
\vspace{-0.5em}

In all conducted experiments, the denoisers correspond to a DRUNet \cite{DRUNet} architecture with 5 blocks, trained to denoise images corrupted by noise levels \( \sigma\). To explicitly control the noise level, a dataset is constructed by augmenting the Berkeley segmentation dataset CBSD68 \cite{CBSD} through random erasing and random cropping operations, resulting in 8,000 images. Specifically, for each image \( x \in \mathcal{D}_\text{CBSD68} \), its noisy measurements are generated as \( y = \mathrm{A}x + \sigma \xi \), where \( \xi \sim \mathcal{N}(0, \mathrm{I}_n) \) and $\mathrm{A}$ depends on the typical task. Unless stated otherwise, as in the section on \( \delta_\text{opt}^2 \) where the noise level \( \sigma \) varies, it is set to \( 0.1 \). 
\vspace{-1em}
\subsection{Optimal Value $\delta_\text{opt}$}
In this section, the results of Section \ref{sec:interpretation_scaling} concerning the optimal scaling parameter \( \delta_\text{opt} \) will be studied experimentally, and its interpretation as an indicator of the training quality of the denoiser \( \mathbf{D} \) will be validated.
Since the denoising task becomes harder to learn at lower noise levels, given fixed training conditions (architecture, data, and training hyperparameters), a denoiser trained on lower noise levels is expected to be less effectively trained. In this experiment, we consider four noise levels, \( \sigma_1 < \sigma_2 < \sigma_3 < \sigma_4 \), and four denoisers, \( \mathbf{D}_1, \mathbf{D}_2, \mathbf{D}_3, \mathbf{D}_4 \), trained under the fixed procedure described earlier. As such, \( \mathbf{D}_1 \) is the least well-trained, while \( \mathbf{D}_4 \) is the best-trained. According to the intuition outlined in Section 3, where the value of \( \delta_\text{opt} \) reflects the training quality, with \( \delta_\text{opt} = 1 \) indicating perfect training, we anticipate the relationship \( 1 < 
\delta_{\text{opt}, 4} <
\delta_{\text{opt}, 3} < \delta_{\text{opt}, 2} < \delta_{\text{opt}, 1} \). To observe this, the function \( \delta \mapsto \mathbb{E}_{X, \xi}\big[\|\mathbf{D}_\delta(x + \sigma \xi) - x\|^2\big] \) (computed as the empirical mean of the norm over the augmented dataset) is plotted for \( \sigma \in \{0.01, 0.07, 0.10, 0.12\} \). The experimental results align with and confirm this intuition as shown in Figure \ref{fig:opt-value-experiment}.
\begin{figure}
    \centering \includegraphics[width=0.7\linewidth]{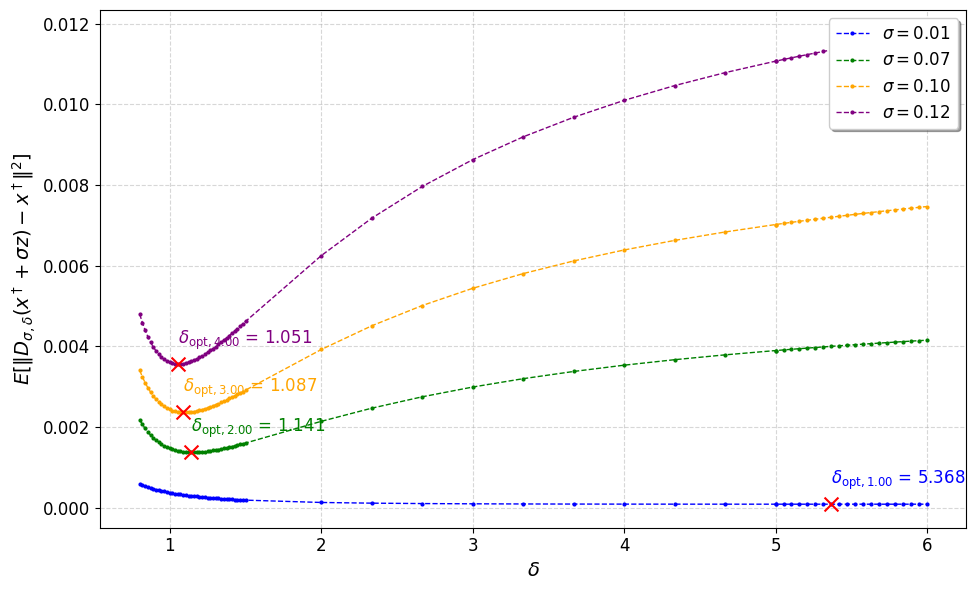}
    \caption{$\delta \mapsto \mathbb{E}_{X, \xi}\big[\|\mathbf{D}_\delta(x + \sigma\xi) - x\|^2\big]$ for $\sigma \in \{0.01, 0.07, 0.10, 0.12\}$}
    \label{fig:opt-value-experiment}
\end{figure}


\vspace{-1.5em}
\subsection{Stability}
In the previous section, the following stability result was stated:  
\begin{align*}
    \text{\textit{Assuming that} } \mathbf{D}_\delta \text{ \textit{is contractive,} } x \mapsto \operatorname{PnP}(\delta, x) \text{ \textit{is continuous for all} } \delta > 0.
\end{align*}
This section provides an empirical evaluation of the continuity of the operator \( x \mapsto \operatorname{PnP}(\delta, x) \). A sequence \( (y_k)_k \) satisfying \( \|y_k - y\| \to 0 \text{ as } k \to +\infty \) is considered. More precisely, a clean image \( x \) is given, and \( y = \mathrm{A}x \) is considered, where \( \mathrm{A} \) corresponds to a physics operator for an inpainting task, where 20\% of the pixels in the image are masked. For all \( k \in \mathbb{N}^* \), \( y_k \) is defined as \( y_k = y + \frac{\sigma}{k}\xi \), where \( \xi \sim \mathcal{N}(0, \mathrm{I}_n) \). $\delta$ is here fixed to its optimal value for $\sigma=0.1$, i.e. as before to $1.09$. The convergence of \( (\operatorname{PnP}(\delta, y_k))_k \) towards \( \operatorname{PnP}(\delta, y) \) is then illustrated in Figure \ref{fig:stability}.

\begin{figure}
    \centering
    \includegraphics[width=0.7\linewidth]{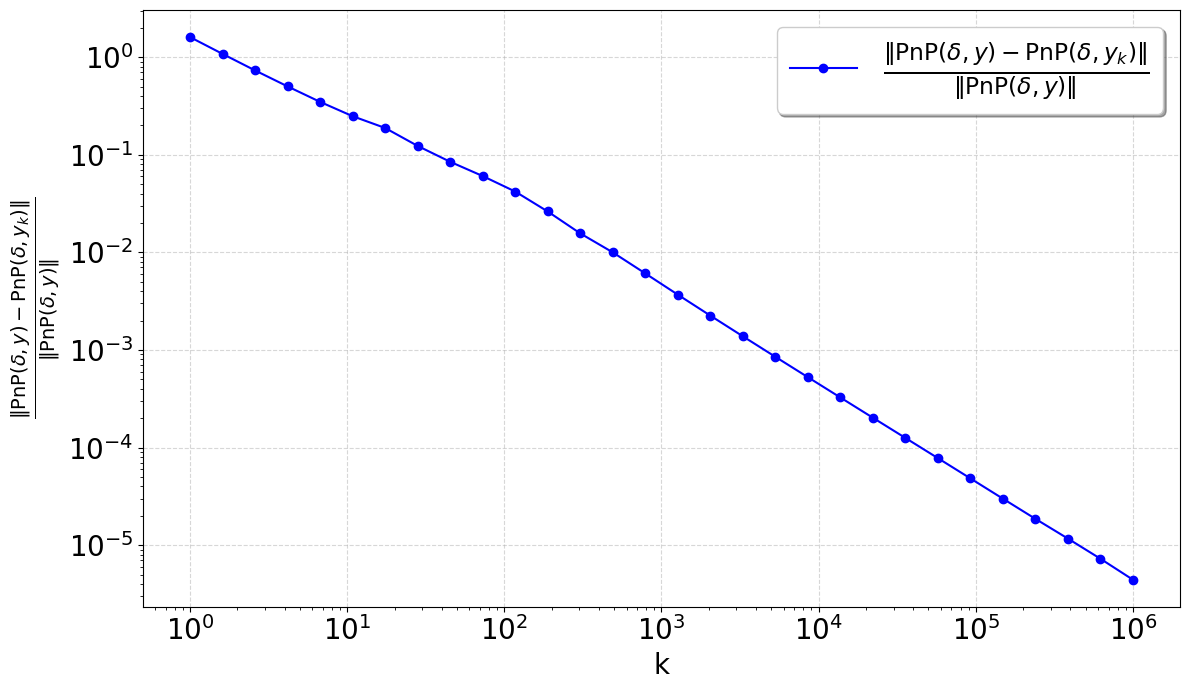} 
    \caption{Convergence of \( (\operatorname{PnP}(\delta, y_k))_k \) towards \( \operatorname{PnP}(\delta, y) \).}
    \label{fig:stability}
\end{figure}
\vspace{-1em}
\subsection{Empirical evidence of convergent regularisation}
\label{sec:empirical_conv_reg}
This section aims to (1) empirically demonstrate that the proposed Tweedie scaling method defines a convergent regularisation for the most basic tasks in PnP-PGD, and (2) show that this is not the case for the method proposed in \cite{denoiser-boosting}. Empirical results show that in the basic cases of denoising ($\mathrm{A} = \mathrm{I}_n$) and of inpainting, no convergence is observed for the homogeneous scaling approach, whereas convergence is observed for Tweedie scaling. The experiment consists of taking a clean image \( x^\dagger \) in the CBSD68 dataset and setting \( y_0 = \mathrm{A}x^\dagger \). A logarithmic range of parameters $\delta$ is then considered, with a step size $\tau = 1$, and $y_0$ is perturbed to $y_\delta = y_0 + \frac{\sigma}{\delta}\xi$, where $\xi \sim \mathcal{N}(0, \mathrm{I}_m)$. For each level of penalisation intensity $\delta$, $x^\dagger_\delta$ is computed as the output of the PnP-PGD algorithm using the Tweedie transformed denoiser $\mathbf{D}_\delta$ and a fixed number of 300 iterations. The initial image \( x_0 \) is a tensor with all values equal to 0 (dark image). Two quantities are plotted in Figure \ref{fig:data-consistency-conv}.
\begin{enumerate}
    \item \textit{Data consistency}: the convergence of $\|\mathrm{A}x^\dagger_\delta - \mathrm{A}x^\dagger\|/\|\mathrm{A}x^\dagger\|\to 0$ as $\delta \to +\infty$, to show that the limit of the sequence $(x^\dagger_\delta)_\delta$ satisfies $\mathrm{A}x^\dagger_\delta = \mathrm{A}x^\dagger$.
    \item \textit{Convergence}: the convergence of $(x^\dagger_\delta)_\delta$ towards a limit $x^*$ to show that there is a notion of convergence.
\end{enumerate}
Regarding the second point, note that the limit $x^*$ need not be the same as the clean image $x^\dagger$ when the measurements are incomplete, i.e.\ $\ker(\mathrm{A})\supsetneq \{0\}$. Indeed, just as in variational regularisation, where there is the notion of a $g$-minimising solution associated with \eqref{regularised-version} \cite{scherzer2009variational}, the regularisation can overcome the non-uniqueness of the solution to the inverse problem, but can only be expected to recover the true clean image if this image is represented well by the prior information encoded in the regularisation.

\begin{figure}
    \centering
    \begin{minipage}[b]{0.48\textwidth}
        \centering
        \begin{subfigure}[b]{1\linewidth}
            \centering
            \includegraphics[height=0.5\textwidth, width=1.0\textwidth]{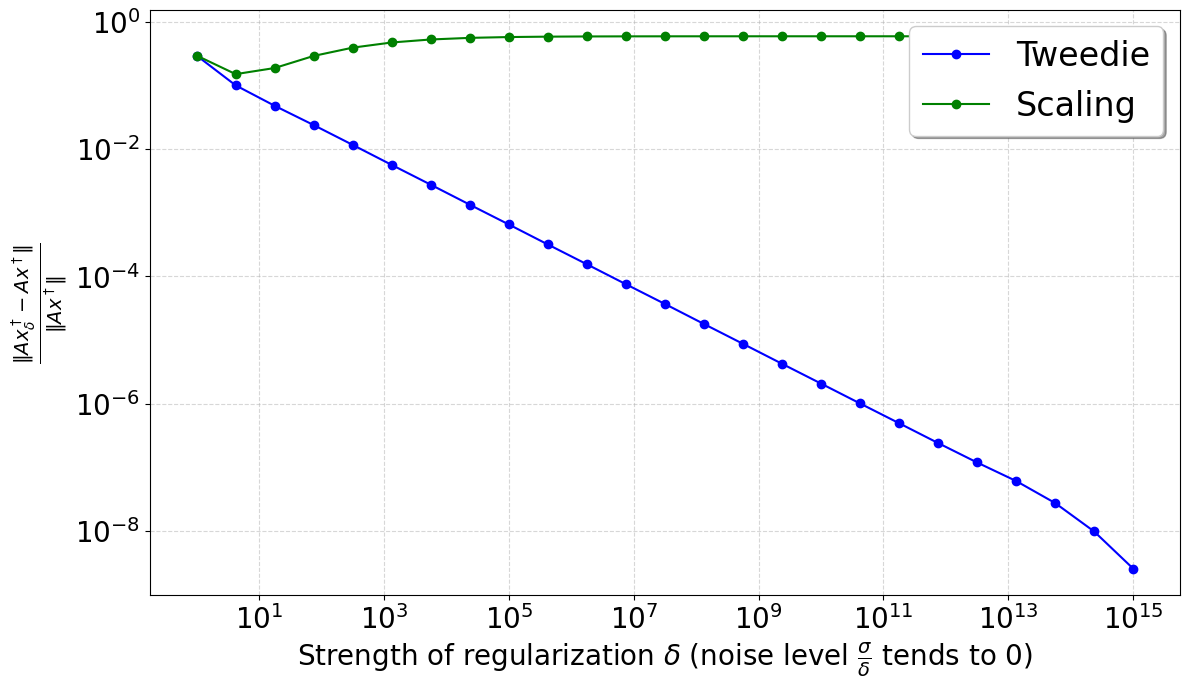}
            \caption{Inpainting: data consistency}
            \label{fig:clean-image}
        \end{subfigure}
        \vspace{0.5em} 
        \begin{subfigure}[b]{1\linewidth}
            \centering
            \includegraphics[height=0.5\textwidth, width=1.0\textwidth]{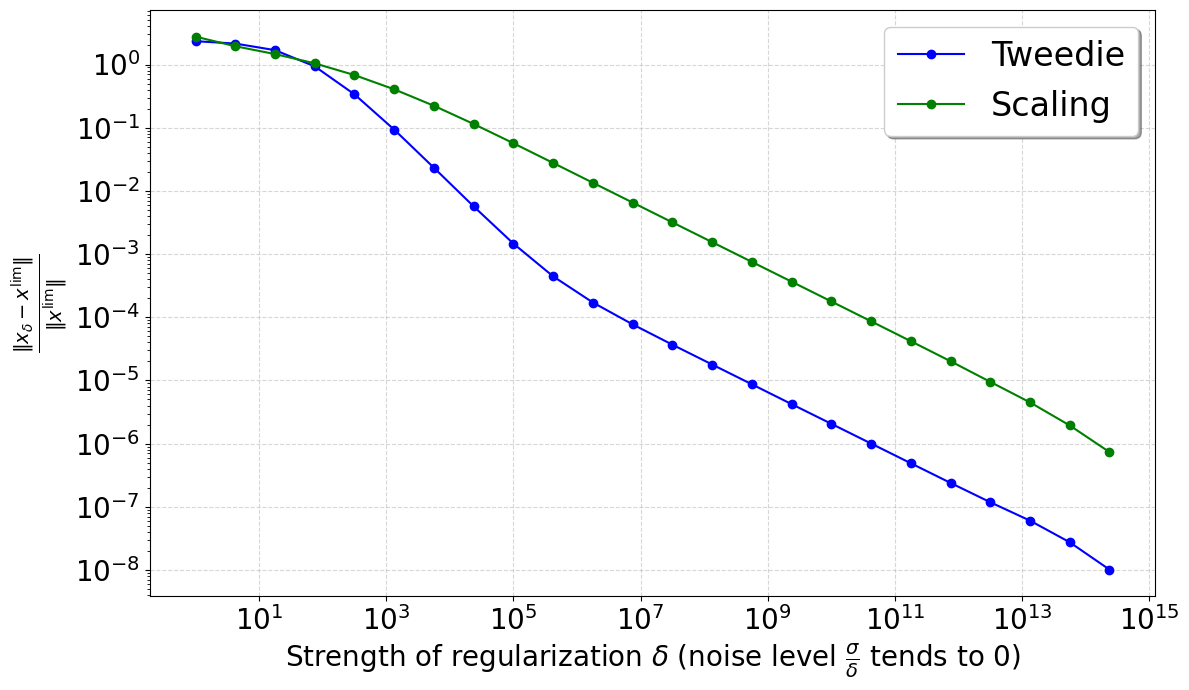}
            \caption{Inpainting: convergence}
            \label{fig:tweedie-denoised}
        \end{subfigure}
    \end{minipage}
    \hfill
    \begin{minipage}[b]{0.48\textwidth}
        \centering
        \begin{subfigure}[b]{1\linewidth}
            \centering
            \includegraphics[height=0.5\textwidth, width=1.0\textwidth]{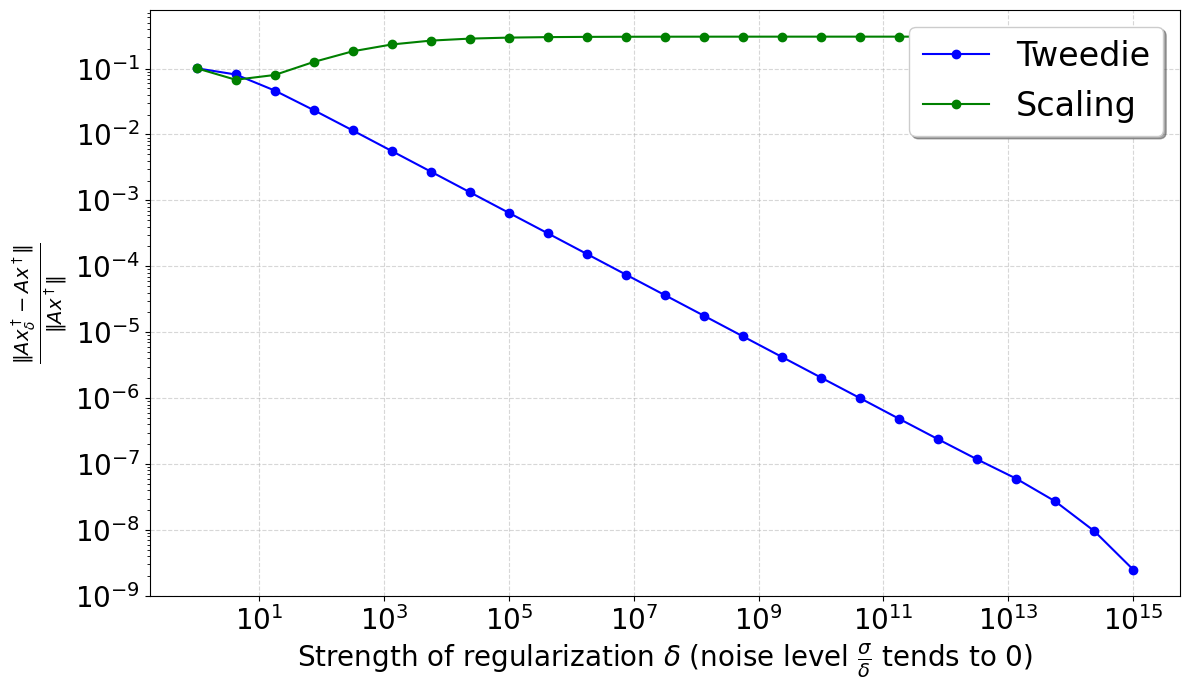}
            \caption{Denoising: data consistency}
            \label{fig:noisy-image}
        \end{subfigure}
        \vspace{0.5em} 
        \begin{subfigure}[b]{1\linewidth}
            \centering
            \includegraphics[height=0.5\textwidth, width=1.0\textwidth]{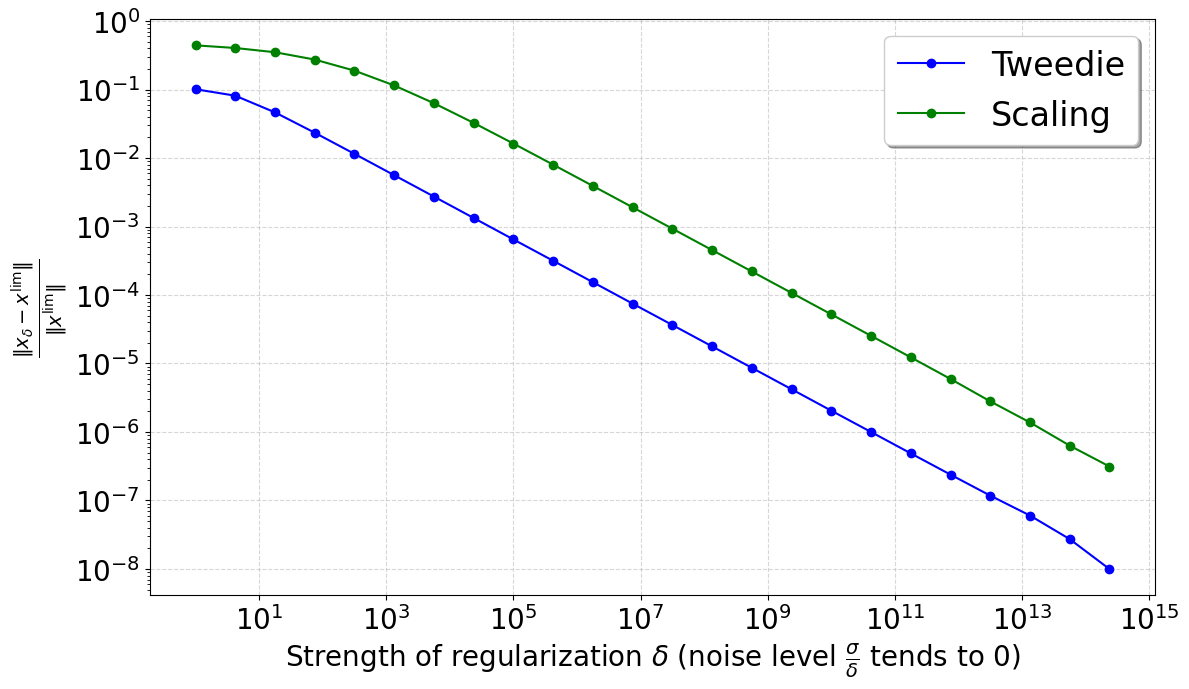}
            \caption{Denoising: convergence}
            \label{fig:scaling-denoised}
        \end{subfigure}
    \end{minipage}
    \caption{Data consistency and convergence (or lack thereof) for the Tweedie scaling method (blue) and homogeneous scaling (green) in Inpainting and Denoising.}
    \label{fig:data-consistency-conv}
\end{figure}
The fundamental premise underlying the construction of the homogeneous scaling method is the existence of a 1-homogeneous functional \( h \) such that \( \mathbf{D} = \operatorname{prox}_h \). Under this assumption, it is evident (by variational regularisation theory) that the  homogeneous scaling family \( (\mathbf{D}_\delta)_\delta \) defines a convergent regularisation. Consequently, if convergence of this regularisation is not observed, it implies that the assumption does not hold for the particular \( \mathbf{D} \) being trained, despite its otherwise generic nature. This observation highlights the limitations inherent in the homogeneous scaling approach.
\vspace{-1.5em}
\section{Conclusion}
\vspace{-0.5em}
The proposed Tweedie scaling method bridges the theoretical framework for convergent regularisation developed in \cite{haltmeier-ebner} by introducing a strongly interpretable scaling parameter. More specifically, it enables principled and explicit modulation of the regularisation intensity in pre-trained denoisers, without relying on restrictive assumptions about the form of the denoiser as in \cite{denoiser-boosting}. The scaling parameter is shown to provide a meaningful characterization of training quality in deep learning-based denoisers. Empirical evidence confirms the framework's ability to guarantee convergence in PnP algorithms. A future direction could involve proposing new scaling techniques to improve the proposed method, including a study of the rate of convergence of the regularisation, investigating whether the property of convergent regularisation holds for other PnP algorithms, and adapting or optimising the scaling parameter according to the algorithms. Other potential research avenues include that of relaxing the assumption of non-expansiveness of the initial denoiser, and considering the infinite-dimensional setting, as is usually the case in the theory of inverse problems.

%
%
%
%

\bibliographystyle{splncs04}
\bibliography{references}
\end{document}